%% file: 00-main.tex
\documentclass[conference]{IEEEtran}

\input{001-recommended-packages}
\input{01-packages}

\begin{document}

\input{02-top-matter} 

\maketitle
\thispagestyle{plain}
\pagestyle{plain}

\input{03-abstract.tex}
\input{10-intro}
\input{20-math}

\input{21-algorithm}

\input{30-experiments}

\input{31-graphprints}

\input{40-conclusion}

\input{90-acks} 

\bibliographystyle{IEEEtran}
\bibliography{refs}

\end{document}

%% file: 001-recommended-packages.tex
\usepackage{cmap}
\usepackage[T1]{fontenc}

\usepackage{graphicx}

\usepackage[ngerman,english]{babel}
\addto\extrasenglish{\languageshorthands{ngerman}\useshorthands{"}}

\usepackage[%
rm={oldstyle=false,proportional=true},%
sf={oldstyle=false,proportional=true},%
tt={oldstyle=false,proportional=true,variable=true},%
qt=false%
]{cfr-lm}
\usepackage[math]{blindtext}

\usepackage{paralist}

\usepackage{csquotes}

\usepackage{microtype}

\usepackage{url}
\makeatletter
\g@addto@macro{\UrlBreaks}{\UrlOrds}
\makeatother

\usepackage{xcolor}

\usepackage{pdfcomment}

\usepackage[capitalise,nameinlink]{cleveref}

\usepackage{xspace}



%% file: 01-packages.tex
\usepackage{wrapfig}
\usepackage{booktabs}
\usepackage{hyperref} 
\usepackage{amsmath, amssymb, amsfonts, amscd} 
\usepackage{amsthm}
\usepackage[justification=raggedright]{subcaption}
\usepackage{courier} 
\usepackage{graphicx}
\graphicspath{ {images/} } 
\usepackage{multicol}
\usepackage{lipsum} 
\usepackage[numbers,sort&compress,square]{natbib} 

\newcommand{\squeezeup}{\vspace{-2mm}} 

\theoremstyle{plain}
\newtheorem{theorem}{Theorem}[section]
\newtheorem{lemma}[theorem]{Lemma}
\newtheorem{corollary}[theorem]{Corollary}
\newtheorem{proposition}[theorem]{Proposition}  

\theoremstyle{definition}
\newtheorem{definition}[theorem]{Definition}

\theoremstyle{remark}

\newcommand{\thm}{\begin{theorem}}
\newcommand{\ethm}{\end{theorem}}
\newcommand{\defn}{\begin{definition}}
\newcommand{\edefn}{\end{definition}}
\newcommand{\eprop}{\end{proposition}}
\newcommand{\prop}{\begin{proposition}}
\newcommand{\cor}{\begin{corollary}}
\newcommand{\ecor}{\end{corollary}}

\IEEEoverridecommandlockouts 

%% file: 02-top-matter.tex
\title{Setting the threshold for high throughput detectors\\
\large{A mathematical approach for ensembles of dynamic, heterogeneous, probabilistic anomaly detectors}
\thanks{This is the extended version of this document including proofs to mathematical results. Please cite the published version appearing in the Proceedings of IEEE Big Data Conference 2017.}  
\thanks{This manuscript has been authored by UT-Battelle, LLC under Contract No. DE-AC05-00OR22725 with the U.S. Department of Energy. The United States Government retains and the publisher, by accepting the article for publication, acknowledges that the United States Government retains a non-exclusive, paid-up, irrevocable, world-wide license to publish or reproduce the published form of this manuscript, or allow others to do so, for United States Government purposes. The Department of Energy will provide public access to these results of federally sponsored research in accordance with the DOE Public Access Plan (\url{http://energy.gov/downloads/doe-public-access-plan}).}
}

\author{
\IEEEauthorblockN{Robert A. Bridges\IEEEauthorrefmark{1}, Jessie D. Jamieson\IEEEauthorrefmark{2}, Joel W. Reed\IEEEauthorrefmark{1}}\\
\IEEEauthorblockA{\IEEEauthorrefmark{1}Computational Sciences \& Engineering Division, Oak Ridge National Laboratory, Oak Ridge, TN\\
\{bridgesra, reedjw\}@ornl.gov}
\IEEEauthorblockA{\IEEEauthorrefmark{2}Department of Mathematics, University of Nebraska, Lincoln, NE, jdjamieson@huskers.unl.edu
}
}

%% file: 03-abstract.tex
\begin{abstract}
Anomaly detection (AD) has garnered ample attention in security research, as such algorithms complement existing signature-based methods but promise detection of never-before-seen attacks. 
Cyber operations now manage a high volume of heterogeneous log data; hence, AD in such operations involves multiple (e.g., per IP, per data type) ensembles of detectors modeling heterogeneous characteristics (e.g., rate, size, type) often with adaptive online models producing alerts in near real time.  
Because of the high data volume, setting the threshold for each detector in such a system is an essential yet underdeveloped configuration issue that, if slightly mistuned, can leave the system useless, either producing a myriad of alerts (and flooding downstream systems) or giving none.  

In this work, we build on the foundations of Ferragut et al. to provide a set of rigorous results for understanding the relationship between threshold values and alert quantities, and we propose a principled algorithm for setting the threshold in practice. 
Specifically, we create an algorithm for setting the threshold of multiple, heterogeneous, possibly dynamic detectors completely a priori, in principle. 
Indeed, if the underlying distribution of the incoming data is known (closely estimated), the algorithm provides provably manageable thresholds. 
If the distribution is unknown (e.g., has changed over time) 
our analysis gives insight into how the model distribution differs from the actual distribution, indicating a period of model refitting is necessary.   
We provide empirical experiments showing the efficacy of the capability by regulating the alert rate of a system with $\approx$2,500 adaptive detectors scoring over 1.5M events in 5 hours. 
Further, we demonstrate on the real network data and detection framework of Harshaw et al. the alternative case, showing how the inability to regulate alerts indicates how the detection model is not a good fit to the data.  
\end{abstract}

%% file: 10-intro.tex
\section{Introduction}
\label{sec:intro}

The current state of defense against cyber attacks is a layered defense, primarily of a variety of automated, signature-based detectors and secondarily via manual investigation by security analysts.
Typically large cyber operations (e.g., at government facilities) have widespread collection and query capabilities for an enormous amount of logging and alert data. 
For example, at the network level, firewalls and intrusion detection/prevention systems (IDS/IPS)  such as Snort\footnote{\url{https://www.snort.org/}} 
produce logs, warnings, and alerts that are collected, in addition to the collection of network flow logs, and sometimes full packet captures are stored and/or analyzed.
\begin{wraptable}{r}{0.23\textwidth}
    \vspace{-.25cm}
    \begin{tabular}{cc}
    \multicolumn{2}{c}{\textbf{\small{Flow Record Example}}}\\
    \toprule    
    \textbf{\small{Time}} & \small{09:58:32.912}\\
    \textbf{\small{Protocol}} & \small{tcp} \\
    \textbf{\small{SrcIP}} & \small{192.168.1.100} \\
    \textbf{\small{SrcPort}} & \small{59860} \\
    \textbf{\small{DstIP}} & \small{172.16.100.10} \\
    \textbf{\small{DstPort}} & \small{80} \\
    \textbf{\small{SrcBytes}} & \small{508526} \\
    \textbf{\small{DstBytes}} & \small{1186562} \\
    \textbf{\small{TotBytes}} & \small{1695088} \\
    \bottomrule
	\end{tabular}
	\caption{Flows record the metadata of IP-IP communications.} 
	\squeezeup
\end{wraptable} 
Additionally, situational awareness tools such as Nessus\footnote{\url{https://www.tenable.com/products/nessus-vulnerability-scanner}} provide lists of software, software version, and known vulnerabilities for each host. 
Host-based IDS/IPS such as  
 McAfee\footnote{\url{https://www.mcafee.com/us/index.html}} 
anti-virus (AV) software 
 and AMP\footnote{\url{http://www.cisco.com/c/en/us/products/security/advanced-malware-protection/index.html}} 
report alerts to cyber security operations, in addition to situational awareness appliances. 
Hence, security analysts now have access to multiple streams of heterogeneous sources producing data in high volumes. 
As an example, a large enterprise network operation, with which we collaborate, monitors only a portion of their network flow logs, a volume of 4-7 GB/s,
in addition to many other logging and alerting tools employed. 
Consequently, manual investigation and automated processing of data/incidents must manage the large bandwidth. 

While the first line of defense is 
signature-based methods (e.g., AV, firewall), which operate by matching precise rules that identify known attack patterns, 
their false negative rate is problematic.
 In response there is a large 
body of literature to use anomaly detection (AD) systems for protection~\cite{ ahmed2007multivariate, axelsson2000base, christodorescu2007can,        ferragut2011automatic,  ferragut2012new, fontugne2010mawilab, garcia2014empirical,   harshaw2016graphprints, lakhina2004diagnosing, moore2017modeling, pevny2012identifying, rehak2009multiagent, scarfone2007guide, sexton2015attack, tandon2009tracking, thomas2010rapid}. 
AD provides a complimentary monitoring tool that holds the promise of detecting novel attacks by identifying large deviations from normal behavior, and this concept has been proven in many of the previous works. 
Ideally, accurate detection with a low number of false positives is achieved. 
\textit{At a minimum, an AD IDS should isolate a manageable subset of events that are sufficiently abnormal to warrant next-step operating procedures}, such as, passing alerts to a downstream system (e.g., automated signature generator) or to an operator for manual investigation. 

While AD has garnered much research attention, such algorithms are met with many challenges when used in practice in the cyber security domain. 
How to design AD models for accuracy\textemdash exploring what statistics, algorithms, and data representations to use so that the detected events correspond with operator-defined positives\textemdash is the focus of many previous works~\cite{ahmed2007multivariate, axelsson2000base, christodorescu2007can,  ferragut2011automatic,  ferragut2012new, fontugne2010mawilab, garcia2014empirical,   harshaw2016graphprints, joslyn2014discrete, sexton2015attack, tandon2009tracking, thomas2010rapid} and in deployment is likely a network-specific task leveraging both domain expertise (understanding attacks, protocols, etc.) and tacit environmental knowledge (understanding configuration of network appliances and their behaviors). 
Common trends to increase accuracy involve the use of ensembles of heterogeneous detectors~\cite{christodorescu2007can, thomas2010rapid, ferragut2011automatic, ferragut2012new, fontugne2010mawilab, garcia2014empirical, kriegel2011interpreting, schubert2012evaluation, sexton2015attack} and/or online detection models that adapt in real time and/or upon observations of data~\cite{ahmed2007multivariate, bridges2015multi, bridges2016multi, ferragut2012new, harshaw2016graphprints}. 
In practice, the need for multiple detectors is enhanced by the diversity in network components (models conditioned on each host, subnet, etc.),  data types (models conditioned on flow data, system logs, etc.), and features of interest (rate, distribution of ports used, ratio of data-in to data-out, etc.).
E.g., patents of Ferragut et al.~\cite{ferragut2016detection, ferragut2016real} detail AD systems using a fleet of dynamic models and producing near real-time alerts on high volume logging data. 

\subsection{Problem Addressed} 
In this work we do \textit{not} present novel methods for accurate detection of intrusions. 
Rather, we address a difficult but important question for AD in IDS applications, namely: 
How should the alert threshold be set in the case of a large number of heterogeneous detectors, possibly changing in real time, that are producing alerts on high-volume data?
Our organization's cyber operations' analysts have arrived at the problem of alert rate regulation from three scenarios that all require real-time prioritization of alerts that can accommodate influxes of data as well as the multitude of evolving models, namely, 
(1) manual alert investigation requires an online way to triage events; 
(2) data storage limitations, e.g., storing
packet captures (PCAPs)  from the most anomalous traffic,  requires a real-time algorithm for prioritizing events  as ``anomalous enough'';
(3) online automated alert processing (e.g.,
automated signature generation of anomalous activity) cannot handle influxes, i.e., downstream systems
require alert rate regulation to prevent a denial-of-service.

To illustrate the problem, consider the AD system and Skaion data used  in Section~\ref{sec:skaion}. 
This AD system scores each flow using two evolving probability models per internal IP, totaling about 2,500 dynamic detectors.  
It is simply not feasible to manually tune the threshold for each model, and even so, since the models are changing in real time, reconfiguration would be periodically necessary. 
Furthermore, the consequences of misconfigured thresholds are substantial. 
Altogether, the system produces a collective $\approx$2M anomaly scores in about five hours; hence, a threshold that is only slightly too high can produce tens of thousands of alerts per hour!  
Moreover, this dataset is small compared to many networks, and the detection ensemble grows linearly with the number of IPs to model (network size). 

The specific problem of how to set the alert threshold for detection systems in these very realistic scenarios is difficult, underdeveloped, and when not properly addressed, leaves anomaly detectors useless, as their goal is to isolate the most abnormal events from the sea of data.  
Hence, we arrive at our problem of interest.  
How does an operator set the threshold for an AD system, given that the method must adequately accommodate a multitude of possibly adaptive models operating on possibly variable speed, high volume data? 
Second, our work contributes to the related problem of detecting drift of adaptive models over time.

\subsection{Background \& New Contributions} 
This alert-rate regulation problem is first specified by Ferragut et al.~\cite{ferragut2012new}, and they note that a principled notion of quantitative comparability across detectors is necessary to deal with multiple/dynamic models. 
Their relevant contribution is twofold. 
(1) By assuming data is sampled from an accessible probability distribution, they formulate a definition equivalent to Defn.~\ref{defn:a-score}. 
The upshot is that anomalies are events with low p-values (see Defn.~\ref{defn:p-value}), and this technique provides a distribution-independent, comparable anomaly score. 
(2) They provide a theorem equivalent to Lemma~\ref{lem:alert-rate} for alert-rate regulation. 
No alert rate algorithm nor experimental testing of the theorem's consequences are presented.

 Kreigel et al.~\cite{kriegel2011interpreting} have addressed the problem of comparing multiple outlier detection methods by manually crafting transformation functions that convert the given output to a score to a comparable output in the interval [0,1].
 This is only done for a handful of outlier detection algorithms, indicating the obvious drawback of this approach, the necessity to manually investigate each model. 

For numerical time-series data, Siffer et al.~\cite{siffer2017ads} exploit the extreme value theorem to find distribution-independent bounds on the rate of extreme (large) values. 

Our contributions build on the work of Ferragut et al.~\cite{ferragut2012new} both in extending the mathematics and in converting these theorems into an operationally-viable solution to the problem. 
New theorems of Section~\ref{sec:math}  provide further mathematical advancements pertinent to understanding the relationship between the p-value threshold and the likelihood of an alert.  
These results informs an operational workflow.
Given (1) a detection capability that uses a probability distribution to score low p-value events as anomalies and (2) knowledge of the data's rate, the operator has a principled, distribution independent method for setting the threshold to regulate the number of alerts produced. 
Hence, the algorithm can be applied to an ensemble of possibly dynamic, heterogeneous detectors to prevent overproduction of alerts. 
Notably, the system will not suppress influxes of anomalies, but asymptotically the operator-given bound is respected. 
Our math results give hypotheses that ensure equality in the theorem;  operationally, this is the case when users can specify, rather than just bound, the number of alerts. 
As the theorems hold independent of the model (distribution) used, operators can set the threshold a priori. 
In particular, it remains valid in a streaming setting, where the detection model is updated in real time to accommodate new observations.
Because the underlying assumption is that future observations are sampled from the model's distribution, the alert rate regulation will fail if the model distribution differs from the actual distribution.
Hence, the theorem's contrapositive gives an operational benefit, namely that violations of the operator-given alert rate indicate that the anomaly  detection model is not a good fit to the data. 
In this case a period of relearning the distribution is necessary for the threshold-setting algorithm to remain effective. 

We present empirical experiments testing our method in setting the alert rate in two scenarios, both using detectors on network flow data. 
The first (Section~\ref{sec:skaion}) shows the efficacy of setting the threshold on approximately 2,500 simultaneous, dynamically updated detectors. 
In this scenario, multiple anomaly scores are computed per flow; hence, the data rate is high and varies according to network traffic.
The results show that the mathematics give a method for regulating the alert rate of dynamic detectors that is a priori (in the sense that no knowledge of the specific distribution is necessary for threshold configuration). 
Our second experiment (Section~\ref{sec:graphprints}) uses data from the network AD paper of Harshaw et al.~\cite{harshaw2016graphprints}, which fits a Gaussian to a vector describing the real network traffic every 30 seconds. 
Hence, it is a single, fixed rate detector. 
The results from this application illustrate the analytic capabilities for gauging model fit that are made possible by the theorems we develop. 

Altogether our work gives new mathematical results regarding the p-value distribution. 
This informs an algorithm that poses an alternative\textemdash Operators can accurately set the threshold of detection ensembles to bound the expected number of alerts or identify a misfit of the detection model.  

%% file: 20-math.tex
\section{Mathematical Results}
\label{sec:math}
In this section we present mathematical assumptions and results that are the foundation for bounding the alert rate of an anomaly detector. 
To proceed, we consider probabilistic anomaly detectors, which score data's anomalousness according to a probability distribution describing the data. 
We leverage the probabilistic description to provide a theorem that gives sharp bounds on the alert rate in terms of the threshold, independent of the distribution. 
This gives an a priori method to manage the expected number of alerts for any distribution. 
As the mathematics is presented with the necessary but possibly abstruse formality and rigor, we include easy-to-understand examples illustrating the results, their implications, and limitations. 

\subsection{Setting and Notation}
\label{sec:notation}
Our setting assumes data is sampled independently from a distribution with probability density function (PDF) $f: X \to [0,\infty)$. 
The ambient space, $(X, \mathfrak{M}, m),$ is assumed to be a $\sigma$-finite measure space with measure $m$ and $\mathfrak{M}$ the set of measurable subsets of $X.$ 
We define the probability of a measurable set $S\in \mathfrak{M}$ to be $P_f(S):= \int_S f dm$; i.e., $P_f$ denotes the corresponding probability measure with Radon-Nikodym derivative $f$. 
For almost all applications, $X$ is either a subset of $\mathbb{R}^n$ with $m$ Lebesgue measure, or $X$ is a discrete space with $m$ counting measure.

We say an anomaly score $A(x)$ respects the distribution $f$ if $A(x) \geq A(y)$ if and only if $f(x) \leq f(y)$\textemdash intuitively, $x$ is more anomalous than $y$ if and only if $x$ is less likely than $y$.  
While this can be attained by simply letting $A = h\circ f$ for a decreasing function $h$ (for example, see Tandon and Chan~\cite{tandon2009tracking} where $A:= -\log_2(f)$ ), the work of Ferragut et al.~\cite{ferragut2012new} notes that such an anomaly score inhibits comparability across detectors. 
That is, because such a definition puts $A$ in one-to-one correspondence with the values of $f$, anomaly scores can vary wildly for different distributions.  
The consequence is that setting a threshold is dependent on the distribution, which is problematic especially for two settings that are common. 
The first is a setting where multiple detectors (e.g., detectors for network traffic velocity, IP-distribution, etc.) are used in tandem as each require different models.
For instances in the literature using cooperative detectors see~\cite{christodorescu2007can, thomas2010rapid, ferragut2011automatic, ferragut2012new, fontugne2010mawilab, garcia2014empirical, kriegel2011interpreting, schubert2012evaluation, sexton2015attack}.   
The second setting is when dynamic models (where $f$ is updated upon new observations) are used, as this requires comparison of anomaly scores over time. 
For examples of streaming detection scenarios see~\cite{ahmed2007multivariate, bridges2015multi, bridges2016multi, ferragut2012new,  harshaw2016graphprints}.  

To circumvent this problem, we follow Ferragut  et al.~\cite{ferragut2012new} by assuming observations are sampled independently from $f$, a PDF, and define anomalies as events with low p-value (as do many detection capabilities).  
For any distribution the p-value gives the likelihood relative to the distribution. 
Hence, it always takes values in $[0,1]$, and in the specific case of a univariate Gaussian is just the two-sided z-score.  
\defn [P-Value]
	\label{defn:p-value}
	The {\it p-value} of $x\in X$ with respect to distribution $f$, is denoted $pv_f: X\to [0,1]$ and is defined as 
	$$ pv_f(x): = \int_{\{t: f(t)\leq f(x) \} } f dm  = P_f(\{t: f(t)\leq f(x) \} ).$$
\edefn 
It is clear from the definition that $pv_f(x) > pv_f(y) $ if and only if $x$ is more likely than $y$, since $f\geq 0$. 
Finally, in order to define an anomaly score, simply compose a decreasing function, say, $h$, with the p-value. 
\defn [Anomaly Score] 
	\label{defn:a-score}
	An anomaly score that respects a distribution $f$ is of the form  $A = h\circ pv_f(x)$, for strictly decreasing $h: [0,1] \to [0,1]$.  
\edefn
Since $h$ can be any strictly decreasing function, $h(x) = (1-x)$ is a natural choice for simply inverting [0,1] so that low p-values (anomalies) get high scores and conversely. 
For the theorems that follow, we use both the p-value threshold denoted by $\beta$, and the  corresponding anomaly score threshold is simply $\alpha:= h(\beta)$.


\subsection{Theorems}
\label{sec:theorems}
In this section we present the mathematical results that make precise the relationship between an anomaly threshold and the likelihood of an alert. 
Theorem~\ref{thm:alert-rate} and ensuing corollaries 
give sharp estimates for bounding the alert rate 
in terms of the threshold.  
	



\begin{lemma}
	\label{lem:alert-rate}
	Let $f$ denote a probability distribution.  For all $\beta \geq 0,$ 
	\begin{align*} 
	P_f(\{x: pv_f(x) \leq \beta\}) & \leq \beta \\
	P_f(\{x: pv_f(x) > \beta\}) & \geq 1 - \beta.
	\end{align*}
	Furthermore, equality holds in both if and only if $\beta = \sup \{x\in X : pv_f \leq \beta\}$.  
\end{lemma}

\begin{proof}
	Suppose for the moment there exists $y\in X$ such that $pv_f(y) = \beta$. 
	Then 
	\begin{align*} 
	\{pv_f \leq \beta\} & = \{x: pv_f(x) \leq pv_f(y)\}\\
						& = \{x: f(x)\leq f(y)\}.
	\end{align*} 
	It follows that 
	\begin{align}
	\label{eqn-sharp} 
	P_f(\{x: pv_f(x) \leq \beta \})	& = P_f(\{x: f(x) \leq f(y) \} ) \nonumber \\
								 	& = pv_f(y) = \beta. 
	\end{align} 
	Hence we have equality in this case, which shows the inequalities are sharp, once proven. 

	To prove the inequality let 
	$r = \sup \{x\in X : pv_f \leq \beta\}.$
	There exists $x_n \in X$ such that $pv_f(x_n) \nearrow r,$ hence 
	$$\{x: pv_f(x) \leq \beta\} = \bigcup \{x: pv_f(x) \leq pv_f(x_n)\}.$$ 
	Since the sets on the right are a nested, increasing family, we have 
	\begin{align*}
	P_f(\{x: pv_f(x) \leq \beta\})& = \lim_n P_f(\{x: pv_f(x) \leq pv_f(x_n)\}) \\ 
		& = \lim_n pv_f(x_n) \text{ \hphantom{aaaaa}\hfill by (\ref{eqn-sharp}) } \\ 
		& = r \leq \beta.
	\end{align*}
	This proves the first inequality, and establishes equality iff $\beta = r = \sup \{x\in X : pv_f \leq \beta\}.$ 
	Finally, 
	$P_f(\{x: pv_f(x) > \beta\}) = 1 - P_f(\{x: pv_f(x) \leq \beta\}) \geq 1 - \beta$.
\end{proof}

\noindent Roughly speaking, the lemma says that if we sample $x$ from distribution $f$, and compute its p-values, $pv_f(x)$, the chance that the $pv_f(x)$ is less than a fixed number $\beta$ is less than or equal to $\beta$. 
The next theorem translates this to the AD setting.

\thm 
	\label{thm:alert-rate}[Alert Rate Regulation Theorem]
	Let $h$ be strictly decreasing so that $A_f = h\circ pv_f$ is an anomaly score that respects the distribution $f$. 
	Let $\alpha$ denote the alert threshold (so $x$ is called ``anomalous'' iff $A_f(x) \geq \alpha$), and set $\beta = h^{-1}(\alpha)$.  
	Let $S \subset X$ be a set of independent samples from PDF $f$. 
	Then the expected number of alerts in $S$ is bounded above by $\beta |S|,$ i.e., 
	$$
	E[\{x\in S: A_f(x) \geq \alpha \}] \leq \beta |S|. 
	$$
\ethm
\begin{proof}
	By definition of $A_f$ and $\beta$, we have 
	\begin{align*}
	E[\{x\in S: A_f(x) \geq \alpha \}]	& = \sum_{x\in S} P_f(\{ A_f(x) \geq \alpha \})\\
										& = \sum_{x\in S} P_f(\{pv_f(x)\leq \beta \})\\
										& \leq \beta |S|, 
	\end{align*}
	with the last inequality provided by the Lemma. 
\end{proof}

\cor
\label{cor:eq-1}
If $pv_f:X\to[0,1]$ is surjective, then equality holds in the preceding theorem, lemma for all $\beta$.
\ecor

\cor
\label{cor:continuous}
If $X$ is a connected topological space, $f$ is not the uniform distribution, and  $pv_f$ is continuous, then $pv_f(X) = [a,1]$ for some $a\in [0,1]$; hence, equality holds in the preceding theorem and lemma for all $\beta \in [a,1)$. 
\ecor

\cor
\label{cor:eq-2}
Suppose $X$ is a topological space and, for all $ y > 0,$ $m \{x : f(x) = y \} = 0.$ Then equality holds in the preceding theorem and lemma. 
\ecor

\begin{proof} 
	Let $x_0, x_1\in X$ such that $0\leq f(x_0) \leq f(x_1)$. Then 
	\begin{align} 
	\label{eqn:cor-estimate}
	0 	& \leq pv_f(x_1) - pv_f(x_0) \nonumber \\
		& = \int_{\{ t: f(t)\in (f(x_0), f(x_1)]  \}} f dm \nonumber \\
		& \leq \| f \|_1 m \{ t: f(t)\in (f(x_0), f(x_1)]  \} \\ 
		& = m \{ t: f(t)\in (f(x_0), f(x_1)]  \} \text{ \hphantom{aaa}\hfill since $f$ a PDF.} \nonumber 
	\end{align}
	By hypothesis, for any $x_1\in X,$
	\begin{align*}
		0 	& = m \{ t: f(t) = f(x_1)\} \\
			& = \bigcap_{ \{ x_0 : f(x_0) < f(x_1)\} } m \{ t: f(t)\in (f(x_0), f(x_1)]  \}\\
			& = \lim_{ f(x_0) \nearrow f(x_1)}  m \{ t: f(t)\in (f(x_0), f(x_1)]  \}  
	\end{align*} 
	by continuity of measures from above. 
	Hence, for all $x_1 \in X$ and $\epsilon > 0$ there exists $x_0 \in X$ such that $m \{ t: f(t)\in (f(x_0), f(x_1)]  \} \leq \epsilon$. 
	
	It follows from Inequality~(\ref{eqn:cor-estimate}) that $ \sup \{x \in X : pv_f(x) < pv_f(x_1)\} = pv_f(x_1).$
	This establishes the condition of Lemma~\ref{lem:alert-rate} for equality with $\beta = pv_f(x_1)$.  
\end{proof}

\subsection{Examples and Explanations} 


See Figure~\ref{fig:f_pvf} depicting a simple trinomial distribution and corresponding p-value distribution.
This distribution's plateau forces a discontinuity in the rate of alerts as a function of the threshold.
\begin{wrapfigure}{r}{0.24\textwidth} 
	\centering
	\includegraphics[scale=0.19]{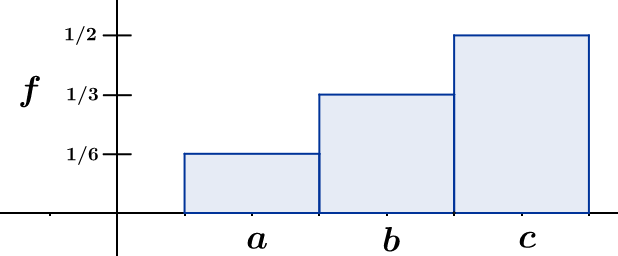}
	\includegraphics[scale=0.19]{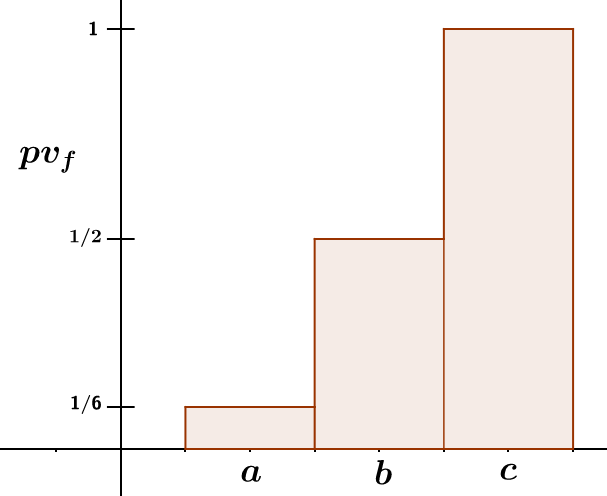}
	\caption[caption]{Trinomial distribution with corresponding p-value distribution. P-value thresholds 1/6, 1/2, and 1 are the only values for which equality holds in the theorems. For these threshold values the expected percentage of alerts are exactly 1/6, 1/2, 1, and, moreover, these are the only percentages possible; e.g., using p-value threshold $\beta\in [0,1/6)$ will yield exactly 0 events and $\beta\in [1/6, 1/2)$ yields an expected 1/6 of the events as alerts. This illustrates a fundamental limitation of PDFs with plateaus. Note that this phenomenon can occur with continuous $f$ as well.}
	\label{fig:f_pvf}
\end{wrapfigure}
In this distribution the operator can either yield exactly none or 1/6\textsuperscript{th} of all events as the threshold changes from below to above $pv_f = 1/6$. 
As the extreme case, consider the uniform distribution in which all events are equally likely/anomalous. 
With the uniform distribution, the operator can yield exactly none or all events. 
Note that the limitation is independent of the method for choosing the threshold and poses a general problem for AD. 
This limitation appears in our experiments with real data. 
	
	

Corollaries~\ref{cor:eq-1}, \ref{cor:continuous}, and \ref{cor:eq-2}  are crafted to identify when this limitation is not present. 
As a simple example, consider the standard normal distribution, Figure~\ref{fig:normal}. 
Regarding the plot of the corresponding p-value distribution is easy to see continuity. 
Since $pv_f(x) = 2F(-|x|)$, where $F$ is the cumulative distribution function (CDF), it follows from Corollary~\ref{cor:continuous} that we have equality in Theorem~\ref{thm:alert-rate} for all $\beta\in[0,1].$ 
The same result follows from Corollary~\ref{cor:eq-2} and the fact that $f$ has no plateaus. 
Hence, the equality condition for all threshold values means that one can specify the expected number of alerts; quite explicitly, if one desires exactly the most anomalous 1/1000\textsuperscript{th} of the data, then simply setting the p-value threshold to $\beta = 0.001$ guarantees the result. 
Using the contrapositive, we see that if the model $f$ admits equality for some p-value threshold, $\beta$, then an average number of alerts above/below $\beta$\% indicates that the tails of $f$ are too small/big, respectively. 
Without equality one can only detect tails that are too thick. 

Finally, we note that while these examples are simple distributions chosen for illustrative purposes, the theorems hold under the specified, very general hypotheses. 
All that is needed is a known measure $m$ for which the probability measure is absolutely continuous.

%% file: 21-algorithm.tex
\subsection{Alert Rate Regulation Algorithm}
\label{sec:algorithm}
Under the assumption that data observations are samples from our distribution, we are mathematically equipped to design an algorithm that exploits the relationship between the alert rate and the threshold to prevent an overproduction of alerts. 
To illustrate this, suppose we receive $N$ data points per time interval $\delta_t$ (e.g., per minute), but operators only have resources to inspect the most anomalous $M \leq N$ in each time interval. 
\begin{wrapfigure}{l}{0.24\textwidth} 
    \vspace{-.1cm}
	\centering
	\includegraphics[scale=0.24]{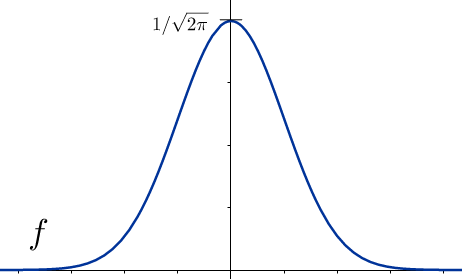}    
	\includegraphics[scale=0.3]{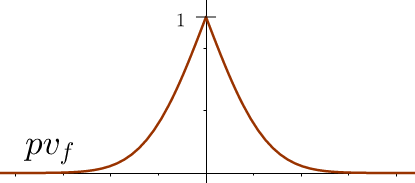}
	\caption[caption]{The standard normal distribution PDF is depicted with corresponding p-value distribution. In this case, the p-value is continuous, and the issue faced by the aforementioned trinomial distribution is avoided. Operationally, this means for any specified percent $p$, a threshold can be set to isolate the most anomalous $p\%$ of the distribution. \newline}
	\label{fig:normal}
	\vspace{-.5cm}
\end{wrapfigure}
Let $f_n$ be a PDF fit to all previous observations $\{x_1, ..., x_{n-1}\}$. 
Following the assumption that the next observation, $x_n$, will be sampled according to $f_n$, define the anomaly score, $A_n := h\circ pv_{f_n}$ where $h$ is a fixed, strictly decreasing bijection of the unit interval.  
Upon receipt of $x_n$ an alert is issued if $A_n(x_{n}) \geq \alpha$.
Equivalently, if $pv_{f_n}(x_n) \leq h^{-1}(\alpha) =: \beta$.   
Finally, we update $f_n$ to $f_{n+1},$ now including observation $x_n$ and repeat the cycle upon receipt of the next observation.  
Leveraging the theorem above, the expected number of alerts per interval is 
$ \sum_{n = 0}^{N-1} P_{f_n}( A_{n}(x_{n}) \geq \alpha ) \leq N h^{-1}(\alpha).$
Hence, choosing $\alpha = h(M/N)$ (equivalently, flagging if the p-value is below $\beta = M/N$) ensures that the operator's bound on the number of alerts will hold on average.

The method above is for a fixed time interval, or for constant rate data. 
As the speed of the data may vary, we now adapt the above method to dynamically change the alert rate to accommodate variable data speed.  
Let $t_i$ denote arrival time of $x_i$, and let $r$ (alerts per second) be the user-desired upper bound on the alert rate (the analogue on $M$).  
Next, for each time interval we periodically estimate the rate of data by letting $r_k = |\{x_i: t_i \in  [ (k-1)*\delta_t , k*\delta_t) \}|/\delta_t$, so that $r_k$ gives the number of observed events over the $k$\textsuperscript{th} $\delta_t$-length interval. 
Hence $r_k$ is a periodically computed, moving average data speed. 
Alternatively, one could compute $r_n$ on each data point's wait time, i.e.,  $r_n := 1 / (t_n-t_{n-1})$, which could experience much more variance.  
Finally, the new threshold at time $n$ shall be given by $\alpha_k = h( r / r_k )$ for the $k$\textsuperscript{th} interval, or equivalently, the p-value threshold is $\beta_k = r/r_k.$  
This new threshold is used for classifying $x_n$ as soon as $x_n$ is observed. 
This algorithm can be called at each iteration of a streaming algorithm to regulate the alert rate. 

Note that this choice of threshold depends only on the rate of the data ($N$ or $r_k$), and the operator's bound on the alert rate ($M$ or $r$). 
In particular, it is independent of the distribution, and therefore can be set a priori,  regardless of the distribution. 
This is especially applicable in a streaming setting, where $f_n$ is constantly changing. 
Next, this is not a hard bound on the number of alerts per day, but rather bounds the alert rate in expectation. 
The operational impact is that if there is an influx of anomalous data, all will indeed be flagged, but on average the alert rate will be bounded as desired. 
Consequently, it is possible to have greater than $M$ alerts in some time intervals. 
Finally, if one has the luxury of performing post-hoc analysis, such an algorithm is not needed; for example, one can prioritize a previous day's alerts by anomaly score, handling as many as possible. 
Yet, if real time analysis is necessary, e.g., only a fraction of the alerts can be processed, stored, etc., then such an algorithm allows real time prioritization of alerts with the bound preset.

The underlying assumption is that data is sampled from the model's probability distribution, so if this assumption fails, for example, if there is a change of state of the system and/or if $f_n$ poorly describes the observations to come, then these bounds may cease to hold. 
This failure gives an ancillary benefit of the mathematical machinery above, namely, that monitoring the actual versus expected alert rate gives a quantitative measure of how well the distribution fits the data.  
Our work to characterize conditions that ensure equality (Corollaries~\ref{cor:eq-1}, \ref{cor:continuous}, and \ref{cor:eq-2}) serve this purpose. 
Under these conditions, the alert rate bound is an actual equality, and deviations from the expected number of alerts over time (both below and above) indicate a poor model for the data.
On the other hand, when the bound is strict (and equality does not hold), only an on-average over-production of alerts will signal a model that does not fit the data. 
Our experiments on real data illustrate this phenomenon as well.

Finally, we note that the adaptive threshold, $\beta = r/r_k$ with $r_k$ the data rate, is conveniently self tuning for variable speed data, it induces a vulnerability. 
Quite simply, an all-knowing adversary with the capability to increase $r_k$  at the time of attack, can force the threshold to zero, to mask otherwise alerted events. 
On the other hand, this is easily parried with a simply fixed-rate detector on the data rate, i.e., modeling the statistic $r_k$, (e.g., a denial of service flooding detector). 
Note that as $r_k$ is computed each interval, the proposed workaround is a fixed rate detector; hence, the alert rate can be regulated without the vulnerability induced.

\subsection{Impact on Detection Accuracy} 
\label{sec:accuracy}
For high volume situations, the adaptive threshold will reduce the number of alerts during influxes of data. 
Consequently, the true/false positive rates (defined, respectively, as the percentage of positives/negatives that are alerts, and hereafter TPR, FPR) will drop, as the number of alerts (numerator) will be reduced with fixed denominator. 
The effect on the positive predictive value (PPV) also known as precision (defined as the percent of alerts that are positives), will depend on the distribution of anomaly scores to the positive events in the data set. 
In particular, if this distribution is uniform, then precision will be unaffected. 
In this case we note that our theorems give a sharp bound (and ability to regulate) the false detection rate (1-PPV).



	



%% file: 30-experiments.tex
\section{Empirical Experiments}
\label{sec:experiments} 
We present experiments testing the fixed-rate and streaming threshold algorithms on two data sets, Skaion (synthetic) and GraphPrints (real) flow data. 

\subsection{Skaion Data \& Detection System}
\label{sec:skaion}
To test the alert rate algorithm we implemented a streaming AD system on the network flow data from the Skaion data set. 
See Acknowledgments~\ref{sec:acks} for details on Skaion data source information.   
This data was, to quote the dataset documentation, ``generated by capturing information from a synthetic environment, where benign user activity and malicious attacks are emulated by computer programs.''  
There is a single PCAP file for benign background traffic and a PCAP file for each of nine attack scenarios.  
We utilized ARGUS\footnote{\url{http://www.qosient.com/}} (the Audit Record Generation and Utilization System), an open source, real-time, network flow monitor, for converting the PCAP information into network flow data. 
As the different PCAP files had mutually disjoint time intervals, we created a single, continuous set of flows by offsetting the timestamps from the 5s20 (Skaion label) attack scenario flows to correspond with a portion of the ``background'' (i.e., non-attack, 5b5 Skaion label)  flows, and then shuffling the ambient and attack data together so they are sorted by time.\footnote{
The 5s20 attack scenario, titled ``Multiple Stepping Stones,'' begins with the attacker scanning internet-facing systems, gaining access to one of them using an OpenSSL exploit, then leveraging this to gain access to several systems behind the firewall.
The attacker's initial scan of the internet-facing systems is not subtle, and therefore, produces a large spike in the AD system. 
}


All together our test data set has 681,220 background traffic flows spanning five hours 37 minutes (337 minutes) with 227,962 flows from the attack PCAP file included from the 227\textsuperscript{th} minute onward.
The attack PCAP file includes approximately 20 minutes of data before the attacker initiates the attack.
This data set includes 6,905 IPs of which 1,246 are internal IPs (100.*.*.*). 
To put this in perspective, Section~\ref{sec:graphprints} uses flows from a real network of (only)  $\approx$50 researchers and created twice the Skaion flow volume in half the time. Skaion data is relatively small.  



We implement a dynamic fleet of detectors roughly based on the Ferragut et al.~\cite{ferragut2016real} patent and currently used in operation. 
Specifically, for each internal IP we implement two detectors. 
The first models the previously observed inbound and outbound private ports, numbered 1-2048 (1-1024 for outbound, 1025-2048 for inbound traffic), using a 2048-bin multinomial, and follows the recent publication of Huffer \& Reed~\cite{huffer2017situational} where it is shown that the role of a host can be characterized by the use of private ports in flow data. 
The second models the producer-consumer ratio (PCR), which is defined as (source bytes - destination bytes)/(source bytes + destination bytes). 
Hence, PCR is a metric describing the ratio, data in : data out, per flow and takes a value in the interval [-1,1]. 
This is modeled by a 10-bin multinomial. 
Initially, all bins (in both models) are given a lone count; notationally, with $k$ bins ($k = 2048$ or $10$) $f_0 (i) = 1/k$ for all $i = 1, 2, .., k$. 
Upon receipt of the $n$\textsuperscript{th} observation, the p-value is computed; 
$ pv_{f_n}(x_n) = \sum f_n(i),$ with sum over $\{i \in {1, ..., k} : f_n(i)\leq f_n(x_n) \}$.  
Finally, the model is updated from $f_n$ to $f_{n+1} $ by simply incrementing the count of the bin observed and the denominator. 
Mathematically, this is the maximum a posteriori multinomial given the previous observations and a uniform prior.

\begin{figure*}
\vspace{-0.1cm}
\centering
	\begin{subfigure}{0.49\textwidth}
  		\includegraphics[width=.98\linewidth]{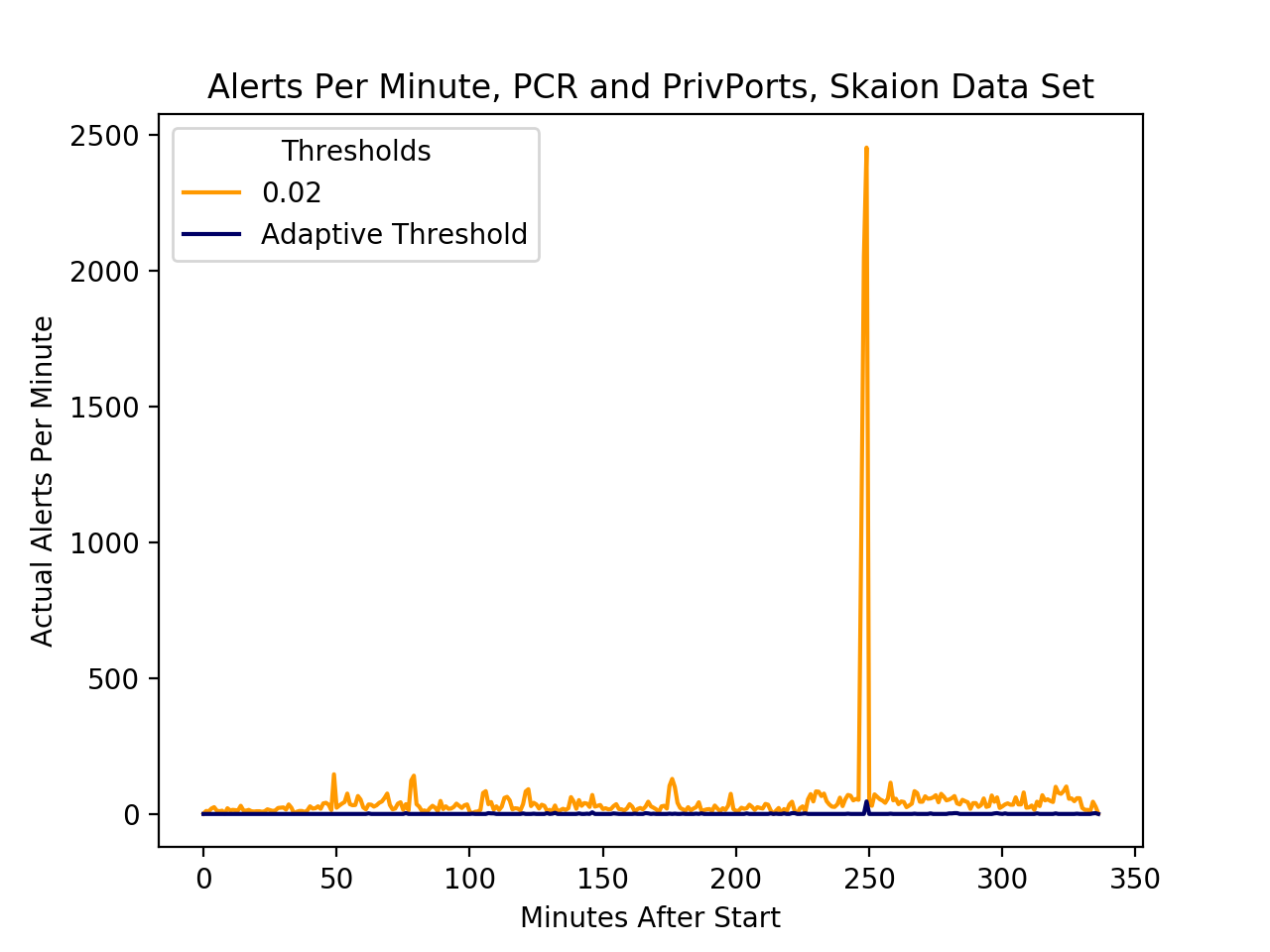}
  		\label{fig:02}
 	\end{subfigure}
 	\begin{subfigure}{0.49\linewidth}
		\includegraphics[width=.98\linewidth]{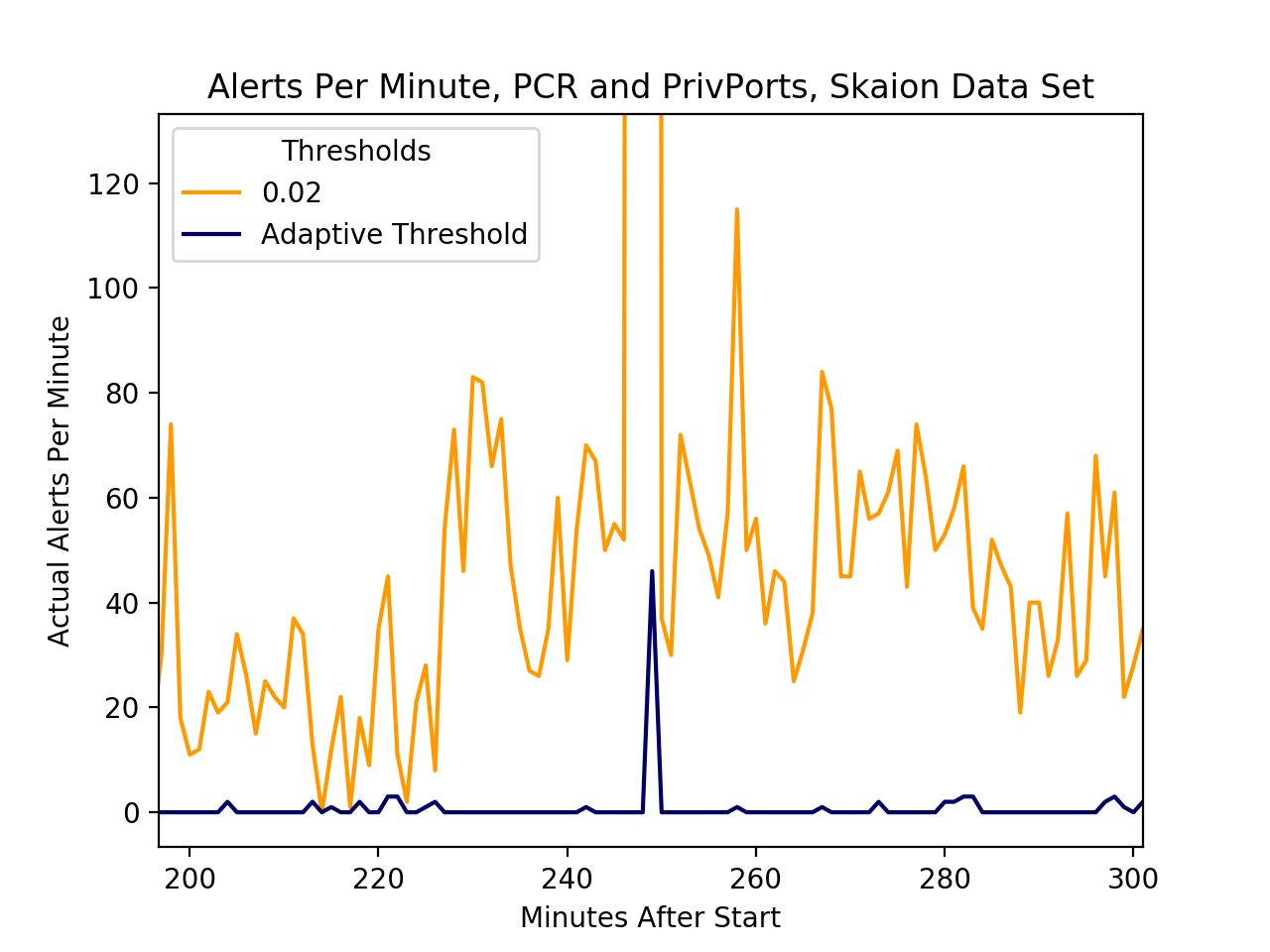}
		\label{fig:02zoom}
	\end{subfigure}
	\begin{subfigure}{0.49\textwidth}
  		\includegraphics[width=.98\linewidth]{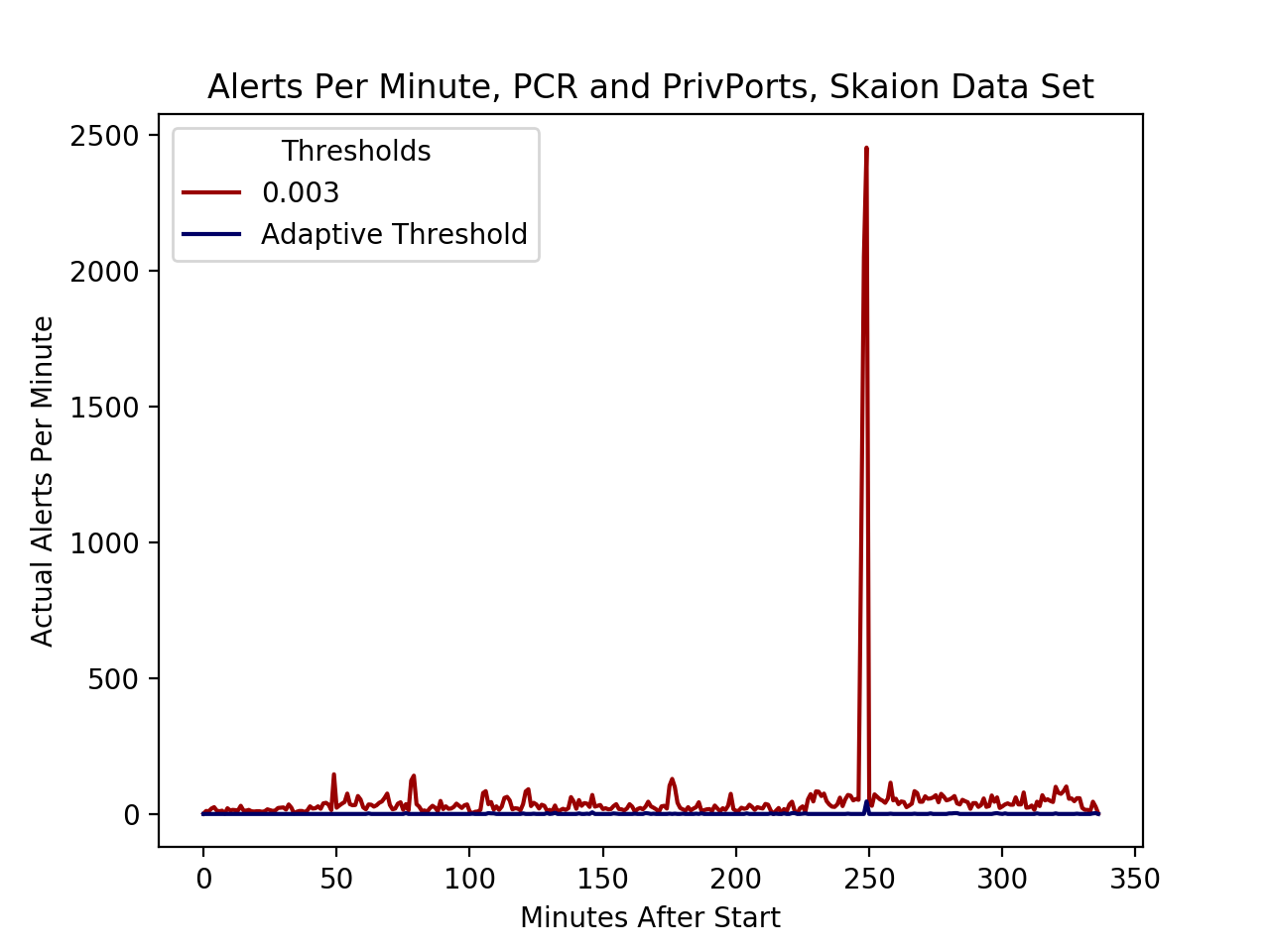}
  		\label{fig:003}
 	\end{subfigure}
 	\begin{subfigure}{0.49\linewidth}
		\includegraphics[width=.98\linewidth]{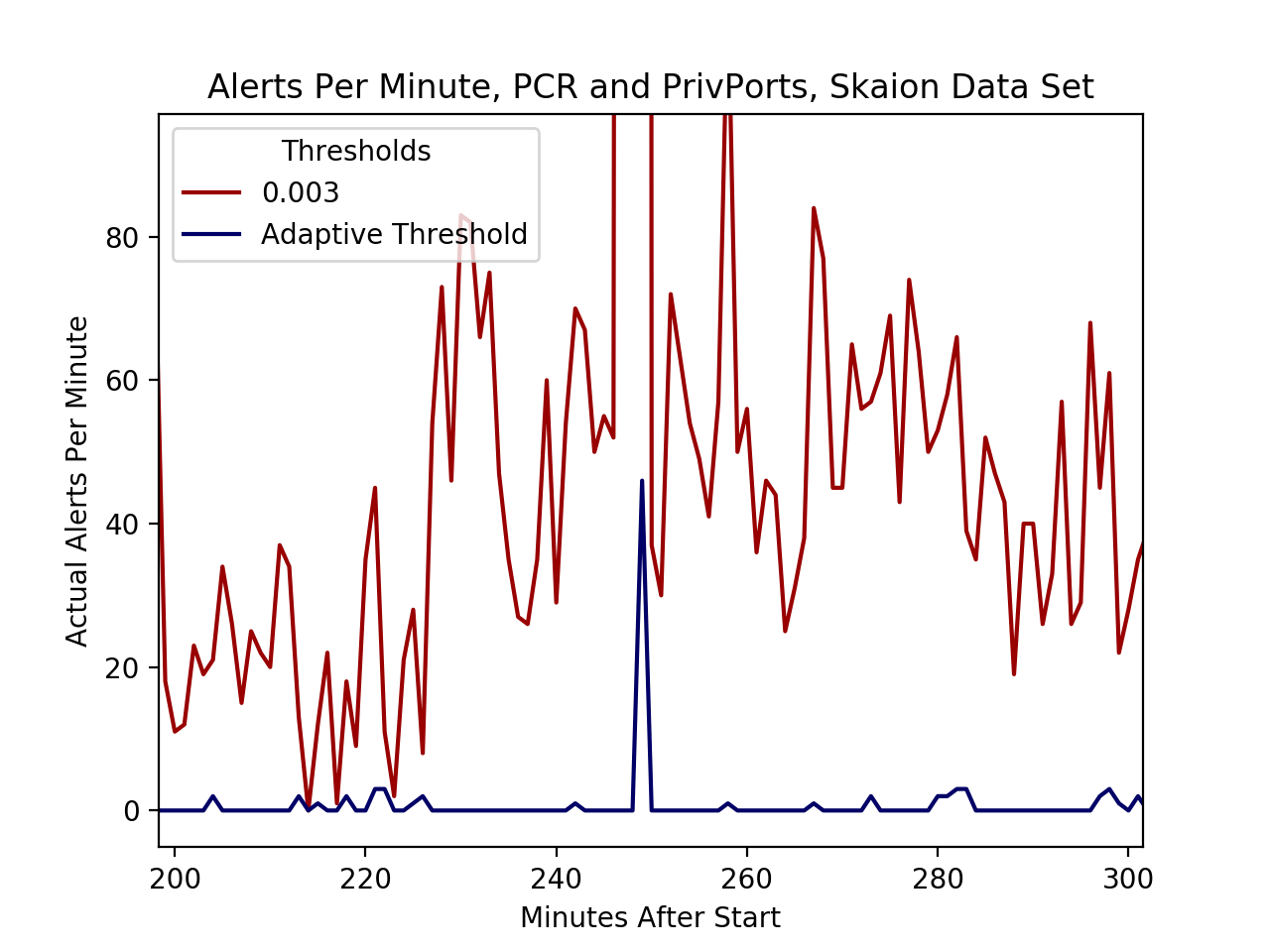}
		\label{fig:003zoom}
	\end{subfigure}
	\begin{subfigure}{0.49\textwidth}
  		\includegraphics[width=.98\linewidth]{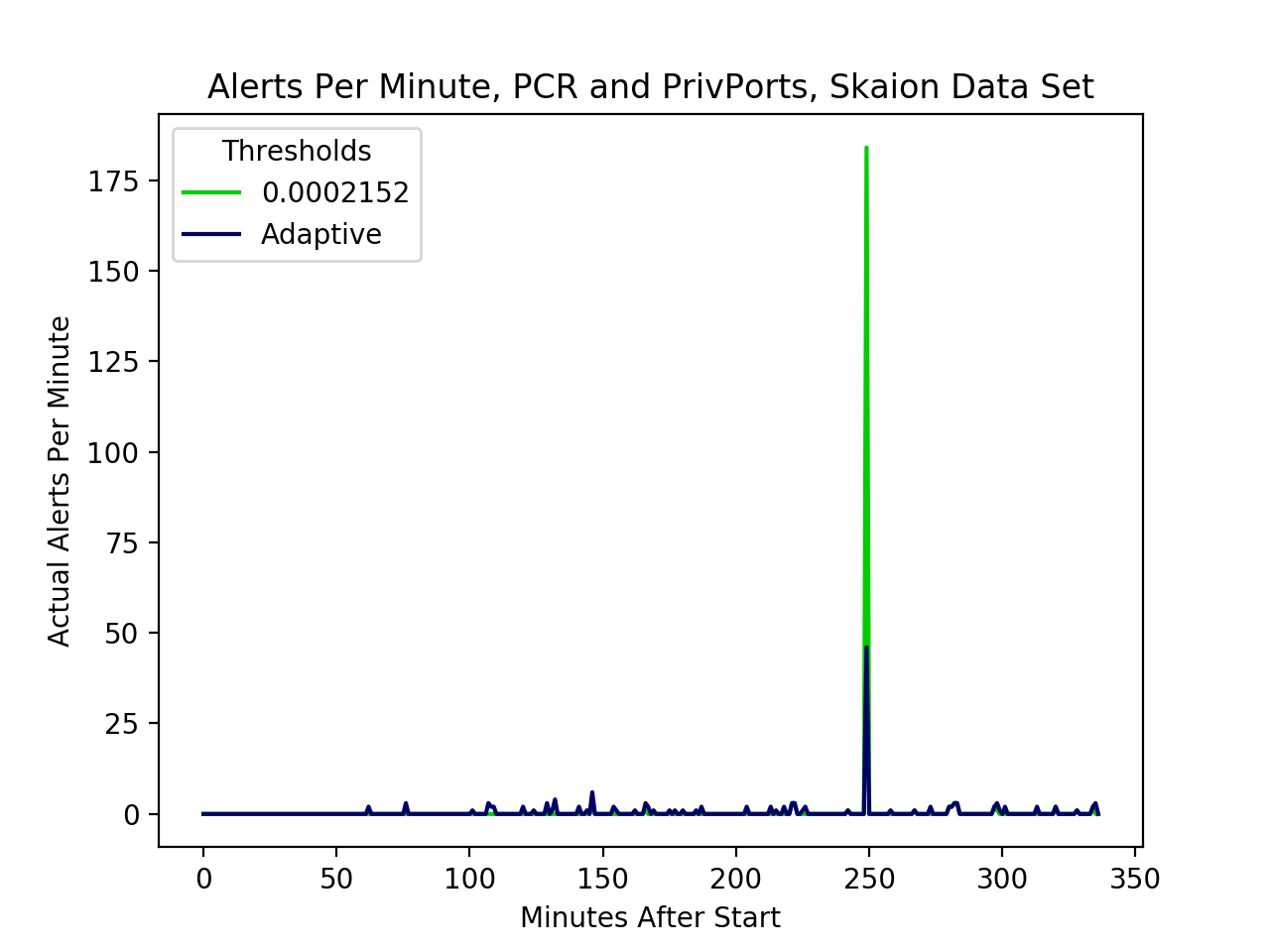}
  		\label{fig:fixed}
 	\end{subfigure}
 	\begin{subfigure}{0.49\linewidth}
		\includegraphics[width=.98\linewidth]{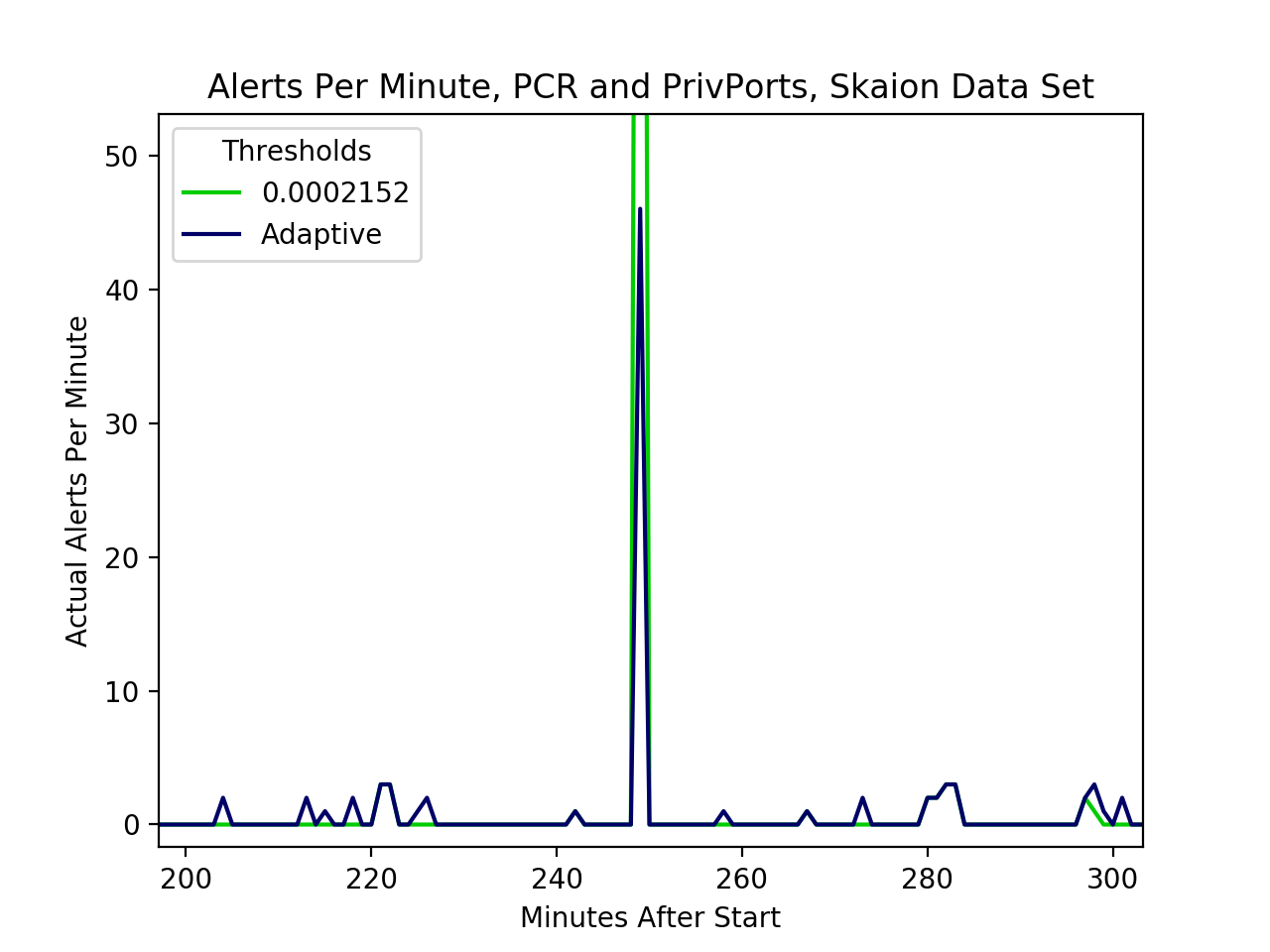}
		\label{fig:fixedzoom}
	\end{subfigure}
\caption{ 
All subplots depict number of alerts per minute ($Y$-axis) against time ($X$-axis) for the adaptive threshold (blue curve in all plots) against a fixed p-value threshold, top row: yellow curve  is $\beta=0.02$, middle row: red curve is $\beta = 0.003$, bottom row: green curve is $\beta = 0.0002152$. 
Left column shows the whole 337 minute dataset; right column is zoomed in to show better resolution. 
Total, the system has 1,246 IPs $\times$ 2 detectors/IP = 2,492 dynamic detectors and  produces 1,565,596 anomaly scores (p-values).
The spike in minute 247 is caused by the attacker initiating a series of unsubtle port scans against internet-facing systems, causing an increased quantity of alerts.   
Fixed thresholds $\beta = 0.02$ and $0.003$ averaging  $\approx50$ and $\approx14$  alerts per minute are used as baselines for comparison. The fixed threshold $\beta = 0.0002152$ is computed using our theorem to satisfy the bound of 1 alert per minute on average and averages $0.64$ alerts per minute. 
Considering the fixed thresholds only,  we conclude that improper choice of the threshold will easily flood the analyst with alerts, but that using a fixed threshold informed by our mathematics will prevent overproduction. 
The adaptive threshold produces 0.43 alerts per minute on average. Considering the bottom row only, we conclude that the adaptive threshold produces less overall alerts than the fixed threshold with the same alert-rate bound, and spreads the alerts more evenly.
}
\label{fig:bigsix}
\end{figure*} 

Altogether, the system has 1,246 IPs $\times$ 2 detectors/IP = 2,492 dynamic detectors and  produces 1,565,596 anomaly scores (p-values) in the 337 minutes of data. 
Our goal in designing these detectors was to create a
\begin{wrapfigure}{r}{0.35\textwidth}
   	\vspace{-.35cm}
    \includegraphics[width=1.1\linewidth]{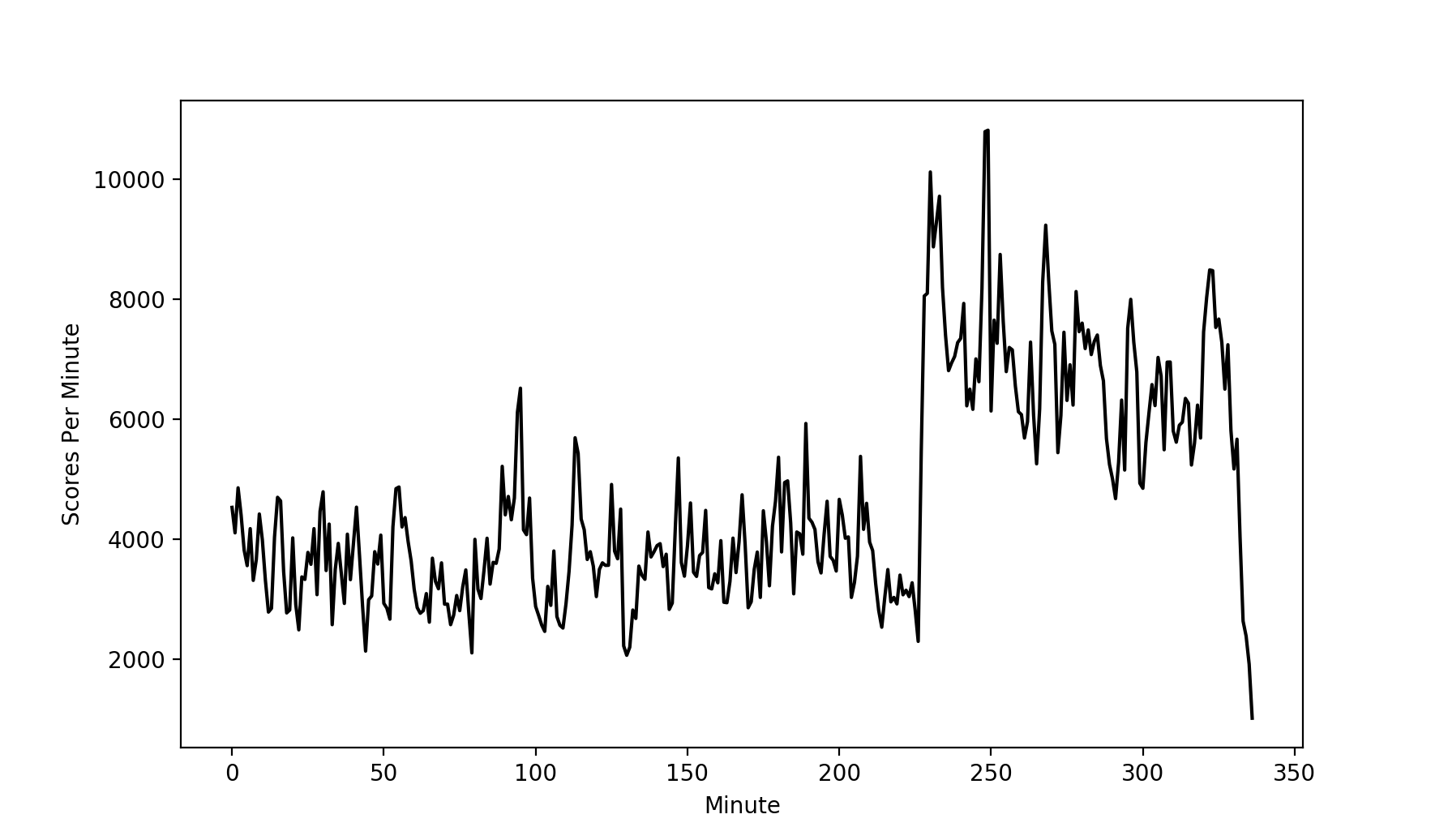}
    \caption{Rate of Skaion flow data, $r_k$, (anomaly scores per minute) plotted. Adaptive threshold (not depicted) has inverse relationship, $r/r_k$, with $r$ the user-given alert rate bound, ($r=1$ in experiments). }
  	\label{fig:r_k}
   	\vspace{-.1cm}
   	\squeezeup
\end{wrapfigure}
realistic detection capability, and models were chosen at the advice of professional cyber security analysts. 
Although the focus of this paper is not on pioneering accurate AD, we quickly note  that the current detectors are able to easily recognize that the initial port scan activity of the attack is anomalous (Fig.~\ref{fig:bigsix}, min. 247). 
Deeper investigation of when and how well these  models are effective is out of scope for this effort. 


\subsubsection{Skaion Data Results}
We test our alert rate threshold analysis against more traditional thresholds. 
Throughout our discussion, reference Figure~\ref{fig:bigsix} giving the number of alerts per minute with various thresholds.  
Some previous works both within intrusion detection \cite{bhaumik2006filtering, cucurull2010disaster,ye2002audit,ye2001chisquare} and in other domains~\cite{patera2008space} using p-value thresholds follow the ``$3\sigma$ rule-of-thumb'' and let 0.3\%= .003 serve as the threshold.  
This follows from the fact that the $3\sigma$-tails of a normal distribution have p-value $0.0026$ or approximately $0.003$.
Others simply choose a value that is near 2\% for an undisclosed reason, or because of a true positive versus false positive analysis in light of labeled events.~\cite{buzen1995masf, campello2015hierarchical, harshaw2016graphprints, lazarevic2003comparative, moore2017modeling, reddy2015acoustic}. 
While such values may be appropriate for hypothesis testing in many situations, for AD on high volume data this unprincipled manner of setting the threshold is inadequate. 
Testing these approaches, we find p-value threshold $\beta = 0.02$  produces an average of $\approx50$ alerts per minute, while $\beta = 0.003$ produces on average $\approx14$ alerts per minute (top two rows of Figure~\ref{fig:bigsix}). The clear conclusion is that the operator will be overwhelmed by alerts with these thresholds. 


We contrast this with the a priori analysis furnished by our theorem. 
Suppose first that our operators can realistically consider $M = 1$ alert per minute on average, and that they know $N\approx 1.5$m scores will be produced.  
With these two figures, we use the theorem to compute the static p-value threshold as follows: 
$1 = E(\mbox{alerts per minute})
= ( 1,565,596/337 )  *P(pv_f(X) \leq \beta)
\leq ( 1,565,596/337 ) *\beta.$
Hence, we set $\beta = M/N = 1/(1,565,596/337) = 0.0002153$. 
This simple calculation shows that the ad-hoc p-value thresholds of $0.003$ and $0.02$ are one and two orders of magnitude too large, respectively. 
Moreover, testing the fixed threshold $\beta = 0.0002153$ yields $\approx 0.64$ alerts per minute. 
See the bottom row of Figure~\ref{fig:bigsix}. 
This shows the efficacy of the alert rate theorem even for fixed thresholds on variable speed data. 

Next, we remove the assumption that the operator knows the number of events ($N$), and  test the dynamically changing threshold with data rate $r_k$ recomputed each minute. 
To do so, we again let the user-defined bound on the average number of alerts be $r = 1$ (following notation of the last Section).  
Each minute we compute the moving average of the data, $r_k$, using the previous minute of data and set $\beta_k= r/r_k = 1/r_k$; i.e., $r_k$ is defined as in Section~\ref{sec:algorithm} with $\delta_t$ = 1 minute. 
See Figure~\ref{fig:r_k} depicting $r_k$. 
Consulting Figure~\ref{fig:bigsix}, we see the adaptive threshold yields about 0-4 alerts per minute except for the one large spike induced by the exceptionally anomalous attack activity. 
On average, the adaptive threshold produces 0.43 alerts per minute, slightly less than the fixed but informed threshold ($\beta = 0.0002152$). 
Finally, consulting the the bottom-row plots of Figure~\ref{fig:bigsix} shows that the distribution of alerts is more spread out with the adaptive threshold.  
Overall, the streaming alert rate regulation algorithm is effective in regulating the nearly 2,500 adaptive detectors with no a priori information. 

\begin{wraptable}{r}{0.3\textwidth}
    \vspace{-.43cm}
    \begin{tabular}{cccc}
    \toprule    
    \textbf{{\small Thresh.}} & \textbf{{\small TPR}} & {\small \textbf{FPR}} & {\small \textbf{PPV}}\\
    {\small .003 }      & {\small .0763 }& {\small .0015} & {\small .7262} \\
    {\small 2.152e-4 }  & {\small .0040 }& {\small 3.48e-5} & {\small .8598} \\
    {\small Adaptive}   & {\small .0021} & {\small 5.33e-5} & {\small .3194} \\
    \bottomrule
    \end{tabular}
	\squeezeup
\end{wraptable} 
Lastly, we present the accuracy metrics for each threshold in the wrapped table. 
As expected, both TPR and FPR both drop when using the adaptive threshold as a result of an overall decrease in alerts. 
Precision (PPV) is significantly lower in the adaptive case. 
To explain this, we regard Figure~\ref{fig:r_k} to see that the data rate nearly doubles on average during the attack times. 
This is caused by the creation of the data set, combining the non-attack Skaion files with the attack Skaion files. 
With the adaptive threshold, the number of alerts will roughly halve during the times when positive examples are included; hence, the drop in precision is an artifact of the simulation.    
In this specific application the adaptive threshold still provides nearly 50 alerts at the onset of the attack. 
More generally, the effect of our threshold algorithm on precision will depend on the distribution of anomaly scores to the attack events.

%% file: 31-graphprints.tex
\subsection{GraphPrints Data \& Detection System} 
\label{sec:graphprints}
\label{sec:exp-graphprints}
We now present experiments of the alert rate developments on the data and network-level anomaly detector (GraphPrints) from the publication of Harshaw et al.~\cite{harshaw2016graphprints}, which was shared with us by the authors. 
In this publication, they chose the tightest threshold that still detects all known positives in the data. 
We will show how our analysis can inform the understanding of the AD model. 


Their work used 175 minutes of network flow data collected from a small office building at our organization, also using ARGUS flow sensor. 
It included flows from approximately 50 researchers using 642 internal IPs plus another $\sim$2,500 internal, reserved IPS (10.*.*.* addresses), and was comprised of 1,725,150 flows. 
The data contained anomalous bittorrent traffic and occasional IP scanning traffic.

The AD method proposed by Harshaw et al., called GraphPrints, uses graph analytics and robust statistics to identify anomalies in the local topology of network traffic. 
The method proceeds in three steps. 
\begin{wrapfigure}{l}{0.25\textwidth} 
    \vspace{-.3cm}
    \includegraphics[width=1.1\linewidth]{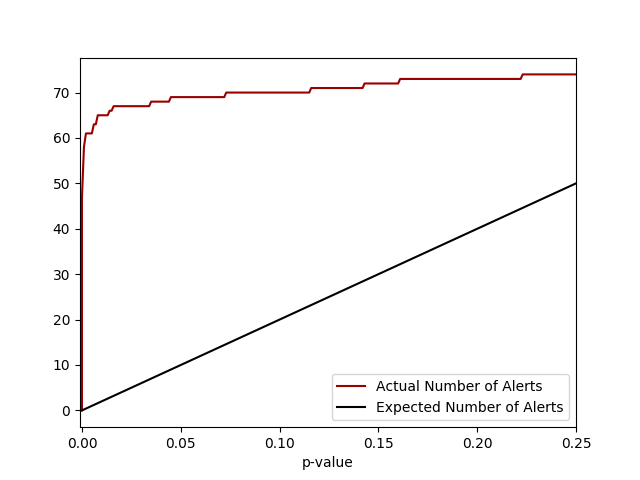}
    \caption{GraphPrints data's actual and expected number of alerts for small p-values. Because a Gaussian distribution was used, the corollaries indicate the two curves should be approximately equal. The extreme disparity in actual vs. expected alerts for small p-values indicates that the data is sampled from a distribution with much thicker tails than the Gaussian used for the AD model. This results in an inability to regulate the number of alerts, indicating problems with their model. }
  	\label{fig:graphprints-analysis}
    \vspace{-.3cm}
\end{wrapfigure}
First, a graph is created for each 31 second interval of network traffic (with one second overlap) by representing IPs by nodes and adding a directed edge from source to destination IP nodes of at least one such flow occurred. 
Moreover, edges are given a binary color indicating if at least one private port occurred in the flow. 
Second, the graph is converted to a graphlet count vector, where each component of the vector is the count of an observed graphlet, which are small, node-induced subgraphs.  
Conceptually, each graphlet can be thought of as a building block graph, and the count vector gives the sums of each such small graph observed in the network traffic graph.\footnote{See Prvzulj et al.~\cite{prvzulj2004modeling} for pioneering work on graphlets.} 
Third, AD is performed by fitting a multivariate Gaussian to the history of observed vectors, and then, given a new vector, alerting if its p-value is below a threshold.\footnote{We note that Harshaw et al. detected events with sufficiently high Mahalanobis distance. 
It is easy to show that Mahalanobis distance is an anomaly score that respects the Gaussian distribution in the sense of Defn.~\ref{defn:a-score} by proving $pv_k(x_0) =  1 - G_k(m(x)^2),$
where $pv_k$ denotes the p-value of the $k-$variate normal distribution, $m$ the Mahalanobis distance, and $G_k$ is the cumulative density of the $\chi^2$ random variable with $k$ degrees of freedom.}
Harshaw et al. implemented this method as a streaming detection algorithm, initially fitting the Gaussian to the first 150 (of 350) data points, and iteratively scoring, then refitting the Gaussian to each of the subsequent  200 events. 
In order to prevent unknown attacks or other anomalies in the data from affecting their model, Harshaw et al. used the Minimum Covariance Determinant (MCD) algorithm with $h = 0.85$. 
This isolates the most anomalous $15\% = (1-h)*100\%$ of the data and fits the Gaussian to the remaining $h = 85\%$.\footnote{See Rousseeuw~\cite{rousseeuw1999fast} for algorithmic details.}

\subsection{GraphPrints Data Results}
First note that because the AD is based on low p-values of a Gaussian distribution, which has no plateaus, Corollaries~\ref{cor:eq-1},\ref{cor:continuous}, and \ref{cor:eq-2} all independently imply that equality holds in Theorem~\ref{thm:alert-rate}. 
Operationally, this means that users can specify the expected number of alerts (not just bound them), provided events are sampled from the model's distribution. 
Testing this for a simple example
 shows alert rate regulation fails; for example, p-value threshold $\beta = 0.01 = 2/200$ corresponds to an expected two alerts in the data, but produces over 60! 
This indicates that our detection model is a bad fit to the data. 
Digging deeper, Figure~\ref{fig:graphprints-analysis} shows that for all small p-values the realized alerts far exceed the expected, violating the theorem. 
We can conclude that the data's distribution has much thicker tails than the Gaussian used for detection\textemdash in short, the detection model is not a good fit to the data. 
This is perhaps unsurprising recalling the use of MCD fitting, which effectively discarded the $15\%$ most outlying observations before computing the mean and covariance. 

The result above illustrates a tradeoff afforded by the mathematical framework\textemdash  either accurate regulation of the alert rate is possible, or the bound\slash equality on the expected number of alerts is not obeyed but information on the fitness (or lack thereof)  of the distribution is gained.

%% file: 40-conclusion.tex
\section{Conclusion} 
In this work we consider the problem of setting the threshold of multiple heterogeneous and/or streaming anomaly detectors. 
By assuming probability models of our data and defining anomalies as low p-value events, we prove theorems for bounding the likelihood of an anomaly. 
Leveraging the mathematical foundation, we give and test  algorithms for setting the threshold to limit the number of alerts of such  a system. 
Our algorithmic developments rely on the underlying assumption that observations are sampled from the model distribution. 
As the theorems hold independently of the distribution, our threshold-setting method persist as models evolve online or for a heterogeneous collection of models (so long as the assumption holds). 
Using the Skaion synthetic network flow data, we implement an AD system of $\approx2,500$ adaptive detectors that scores over 1.5m events in 5 hours, and show empirically how to set the threshold and regulate the number of alerts.  
The mathematical contrapositive of our main theorem operationally provides the user with an alternative\textemdash either the alert rate regulation is possible or the detector's model is a bad fit for the data. 
We demonstrate the use of this analytical insight by implementing the threshold algorithm on the real network data of Harshaw et al. and proving that their data is sampled from a distribution with much thicker tails than their detection model's. 
In summary, our work provides a mathematical foundation and empirically verified method for configuring anomaly detector thresholds to accommodate the hardships necessitated by modern cyber security operations.

%% file: 90-acks.tex
\section*{Acknowledgements}
\label{sec:acks} 
Thank you J. Laska, V. Protopopescu, M. McClelland, L. Nichols, J. Gerber, and reviewers whose comments helped polish this document.  
This material is based on research sponsored by the U.S. Department of Homeland Security (DHS) under Grant Award Number 2009-ST-061-CI0001, DHS VACCINE Center under Award 2009-ST-061-CI0003, and Laboratory Directed Research and Development Program of Oak Ridge National Laboratory, managed by UT-Battelle, LLC, for the U. S. Department of Energy, contract DE-AC05-00OR22725.  
The views and conclusions contained herein are those of the authors and should not be interpreted as necessarily representing the official policies or endorsements, either expressed or implied, of the DHS.
This material is based upon work supported by the National Science Foundation Graduate Research Fellowship Program under Grant No. 25-0517-0143-002. Any opinions, findings, and conclusions or recommendations expressed in this material are those of the authors and do not necessarily reflect the views of the National Science Foundation. 
The data used in this research and referenced in this paper was created by Skaion Corporation with funding from the Intelligence Advanced Research Project Agency, via \url{www.impactcybertrust.org}.